\documentclass[a4paper,oneside, 11pt]{article}
\usepackage{orange}
\usepackage{amssymb,amsfonts,amsmath, bbm,amstext}
\usepackage{graphicx,times,latexsym,makeidx,xspace,url,comment,subfigure,cite, multirow}
\usepackage{algorithm, algorithmic}
\usepackage{color}
\usepackage{epstopdf}

\newcommand{\I}[1]{  \mathbbm{1}_{ \left\{  #1 \right\} }   }
\newcommand{\N}[1]{\mathbbm{N}}

\newcommand{\A}{ {\mathcal{A}} }
\newcommand{\Q}{ {\mathcal{Q}} }

\newtheorem{theorem}{theorem}[section]
\newtheorem{prop}[theorem]{Proposition}
\newtheorem{lem}[theorem]{Lemma}
\newtheorem{cor}[theorem]{Corollary}

\newtheorem{obs}[theorem]{Observation}

\begin{document}

\title{On the impact of TCP and per-flow scheduling on Internet performance (extended version)}

\author{Giovanna Carofiglio$^1$ and Luca Muscariello$^2$\\
$^1$ Bell Labs, Alcatel-Lucent, France, $^2$ Orange Labs Paris\\}

\email{giovanna.carofiglio@alcatel-lucent.com,  luca.muscariello@orange-ftgroup.com}

\version{1.0 ORC1}

\date{31/07/2009}


\maketitle

\abstract{
Internet performance is tightly related to the properties of TCP and UDP protocols,
jointly responsible for the delivery of the great majority of Internet traffic.
It is well understood how these protocols behave under FIFO queuing
and what the network congestion effects. However, no comprehensive analysis
is available when flow-aware mechanisms such as per-flow scheduling and dropping policies
are deployed. Previous simulation and experimental results leave a number of unanswered questions.
In the paper, we tackle this issue by modeling via a set of fluid non-linear ODEs
the instantaneous throughput and the buffer occupancy of $N$ long-lived TCP sources under three per-flow scheduling
disciplines (Fair Queuing, Longest Queue First, Shortest Queue First) and with longest queue drop buffer management.
We study the system evolution and analytically characterize the stationary regime:
closed-form expressions are derived for the stationary throughput/sending
rate and buffer occupancy which give thorough understanding of short/long-term fairness for TCP traffic.
Similarly, we provide the characterization of the loss rate experienced by UDP flows
in presence of TCP traffic.
As a result, the analysis allows to quantify benefits and drawbacks related to the deployment of
flow-aware scheduling mechanisms in different networking contexts.
The model accuracy is confirmed by a set of $ns2$ simulations and by the evaluation
of the three scheduling disciplines in a real implementation in the Linux kernel.}

\section{Introduction}\label{sec:intro}

\subsection{Congestion control and per-flow scheduling}
Most of the previous work on rate controlled sources, namely TCP, has considered networks
employing FIFO queuing and implementing a buffer management scheme like drop tail or AQM (e.g. RED \cite{misra}).\\
As flow-aware networking gains momentum in the future Internet arena (see \cite{nick09}), per-flow scheduling
already holds a relevant position in today networks:  it is often deployed in
radio HDR/HSDPA \cite{umts}, home gateways \cite{ostallo,sqf-demo}, border IP routers \cite{hpsr,ancs05},
IEEE 802.11 access points \cite{tan}, but also ADSL aggregation networks.
Even if TCP behavior has been mostly  investigated under FIFO queuing, in
a large number of
significant network scenarios the adopted queuing scheme is not FIFO.\\
It is rather fundamental, then, to explore the performance of TCP in these less studied
cases and the potential benefits that may originate by the application of per-flow scheduling
in more general settings, e.g. in presence of rate
uncontrolled sources, say UDP traffic.\\
In this paper we focus on some per-flow schedulers with applications in wired networks:
\begin{itemize}
\item
\textit{Fair queuing} (FQ) is a well-known mechanism to impose fairness into a network link and it has
already been proved to be feasible and scalable on high rate links (\hspace{-0.1mm}\cite{suter, korte04, korte05}).
However, FQ is mostly deployed in access networks, because
of the common belief that per-flow schedulers are not scalable, as the number of running flows
grows in the network core.
\item
\textit{Longest Queue First} (LQF).
In the context of switch scheduling for core routers with virtual output queuing (VOQ),
throughput maximization has motivated the introduction of
an optimal input-output ports matching algorithm, maximum weight matching (MWM, see \cite{nick}),
that, in the case of a N-inputs-1-output router, reduces to selecting longest queues first.
Due to its computational complexity, MWM has been replaced by a number of heuristics, all equivalent
to a LQF scheduler in a multiplexer.
All results on the optimality of MWM refer to rate uncontrolled sources,
leaving open questions on its performance in more general traffic scenarios.

\item
\textit{Shortest Queue First} (SQF).
The third per-flow scheduler under study is a more peculiar and less explored scheduling discipline
that gives priority to flows generating little queuing.
Good properties of SQF  have been experimentally observed in \cite{ostallo,sqf-demo} in the context of home gateways
regarding the implicit differentiation provided for UDP traffic.
In radio access networks, as HDR/HSDPA, SQF has been shown, via simulations, to improve TCP completion times (e.g.\cite{umts}).
However, a proper understanding of the interaction between SQF and TCP/UDP traffic still lacks.

\end{itemize}

Multiple objectives may be achieved through per-flow scheduling such as fairness,
throughput maximization, implicit service differentiation. Therefore, explanatory models are necessary
to give a comprehensive view of the problem under general traffic patterns.
\subsection{Previous Models}
A vast amount of analytical models is behind the progressive understanding of the many facets of Internet congestion control
and has successfully contributed to the solution of
a number of networking problems. To cite a few, fair rate allocation \cite{Kelly98, massoulie}, TCP throughput evaluation
(see \cite{padhye,misra,altman,baccelli,ajmone}) and maximization (Split TCP \cite{carofiglioSplit}, multi-path TCP \cite{multipathTCP}),
buffer sizing \cite{Wischik}, etc.
Only recently, some works have started modeling
TCP under per-flow scheduling in the context of switch scheduling (\hspace{-0.2mm}\cite{giaccone, keslassyTCP}),
once made the necessary distinction between per-flow and switch scheduling, the latter being aware of input/output ports and not
of single user's flows.
More precisely, the authors focus on the case of one flow per port, where both problems basically fall into the same.
In \cite{giaccone} a discrete-time representation of the interaction between
rate controlled sources and LQF scheduling is given,
under the assumption of a per-flow RED-like buffer management policy that distributes early losses
among flows proportionally to their input rate (as in \cite{fred, afd}).
Packets are supposed to be chopped in fixed sized cells as commonly done
in high speed switching architectures.
\\In \cite{keslassyTCP} a similar discrete-time representation of scheduling dynamics is adopted for LQF and FQ,
with the substantial difference of separate queues of fixed size, instead of
virtual queues sharing a common memory.
Undesirable effects of flow's stall are observed in \cite{keslassyTCP}
and not in \cite{giaccone} due to the different buffer management policy.
Indeed, the assumption of separate physical queues is not of minor importance, since it may lead to flow  stall
as a consequence of tail drop on each separate queue.
Such unfair and unwanted effects
can be avoided using virtual queues with a shared memory jointly with \textit{Longest Queue Drop} (LQD)  buffer management
(see\cite{suter}).
Such phenomenon is known and it has been noticed for the first time in \cite{suter}.
In this way, in case of congestion, sources with little queuing are not penalized by packet drops,
allocated, instead, to more greedy flows.
Both works (\hspace{-0.1mm}\cite{giaccone},\cite{keslassyTCP}) are centered on TCP modeling, while considering UDP
(\hspace{-0.1mm}\cite{keslassyTCP}) only
as background traffic for TCP without any evaluation of its performance.
Finally, none of the aforementioned models has been solved analytically,
but only numerically and compared with network simulations.
\subsection{Contribution}
The paper tackles the issue of modeling the three above mentioned
flow-aware scheduling disciplines in presence of TCP/UDP traffic in a comprehensive
analytical framework based on a fluid representation of system dynamics.
We first collect some experimental results in Sec.\ref{sec:exp-rem}.
In order to address the questions left open,
we develop in Sec.\ref{sec:model} a fluid deterministic model
describing through ordinary differential equations either TCP sources behavior either virtual
queues occupancy over time.
Model accuracy is assessed in Sec.\ref{sec:accuracy} via the comparison against $ns2$ simulations.
In presence of $N=2$ TCP flows, the system of ODEs presented in Sec.\ref{sec:model} is
analytically solved in steady state and
closed-form expressions for mean sending rates and throughputs are provided in Sec.\ref{sec:Analytical-results}
for the three scheduling disciplines under study (SQF, LQF and FQ).
The model is, then, generalized to the case of $N>2$ flows and
in presence of UDP traffic in Sec.\ref{sec:extensions}.
Interesting results on the UDP loss rate are derived in a mixed TCP/UDP scenario.
A numerical evaluation of analytical formulas  is carried out in Sec.\ref{sec:numerics} in comparison with packet-level simulations.
Finally, Sec.\ref{sec:discussion} summarizes the paper contribution and sheds light on potential applications of per-flow scheduling and
particularly SQF.

\section{Experimental remarks}\label{sec:exp-rem}
In this section we consider a simple testbed as a starting point of our analysis.
An implementation of FQ is available in Linux and, in addition, we have developed
the two missing modules needed in our context, LQF and SQF.
All three per-flow schedulers have similar implementations: a common memory is shared by virtual queues,
one per flow.
Packets belong to the same flow if they share the same
5-tuple (IP src and dst, port numbers, protocol) and in case of memory
saturation the flow with the longest queue gets a packet drop (LQD).
Our small testbed is depicted in Fig.\ref{fig:graph-test} where sender and receiver
employ the Linux implementation of TCP Reno. Different round trip times (RTTs) on each flow are obtained adding
emulated delays through netem and iptables (\hspace{-0.2mm}\cite{lartc}).
Three TCP flows with different RTTs  of 10, 100 and 200ms are run in parallel in the testbed,
and their long term throughput is measured at the receiver. Each test lasts 5 minutes and
the throughput is averaged over 10 runs. The result is reported in Fig.\ref{fig:graph-test}
but they can also be seen through the Jain index fairness J (\hspace{-0.2mm}\cite{jain})
in Tab.\ref{tab:fairness}.
About the long term throughput, we observe that
FQ is fair and the throughput is not affected by RTTs. On the contrary, LQF and SQF
suffer from a RTT bias: LQF favors flows with small RTT, while SQF favors flows with large RTT.
Moreover Tab.\ref{tab:fairness} shows that, while FQ and LQF show no difference between long and short term,
SQF is much more unfair at short time scales as J reduces to 0.55 in the short term, while being  0.735 in the long term
(J=1/N means that one flow out N gets all the resource over the
time window, J=1 is for perfect sharing).
\begin{figure}[htbp]
\centering
$\begin{array}{ccc}
\includegraphics[width=0.2\textwidth, height=0.3\textwidth]{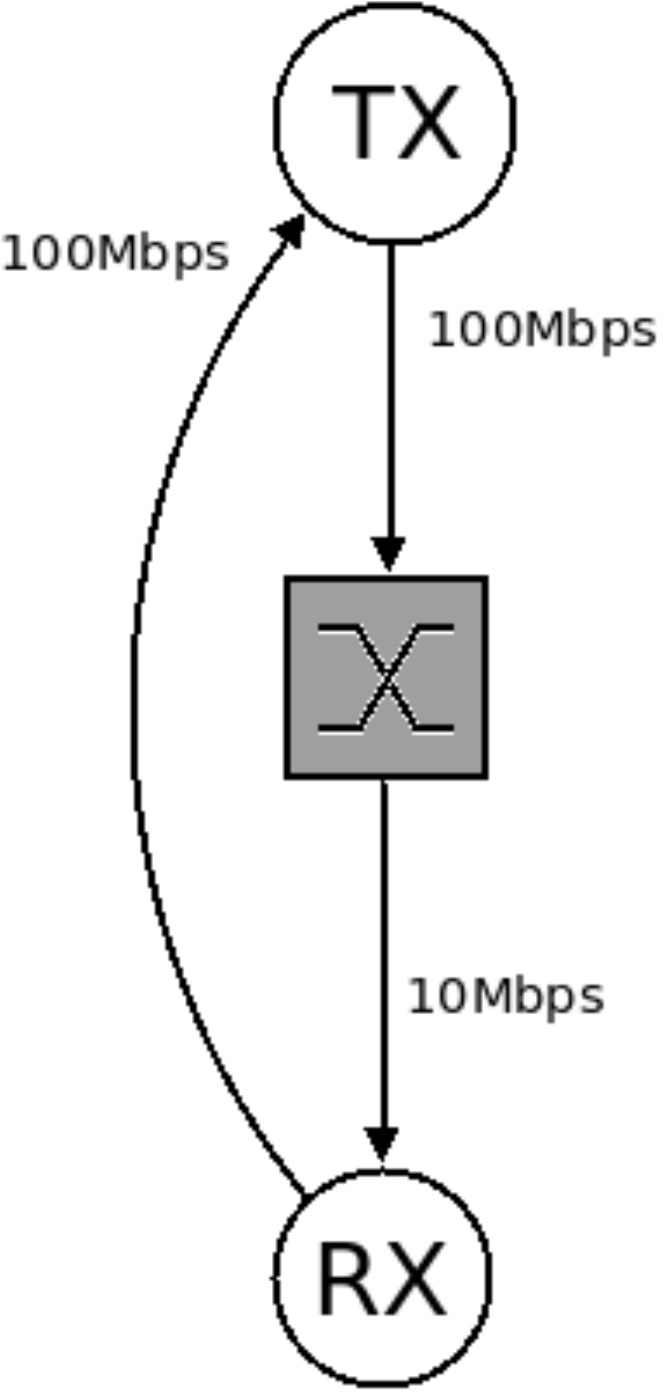}\hspace{4mm}
&
\includegraphics[width=0.4\textwidth, height=0.3\textwidth]{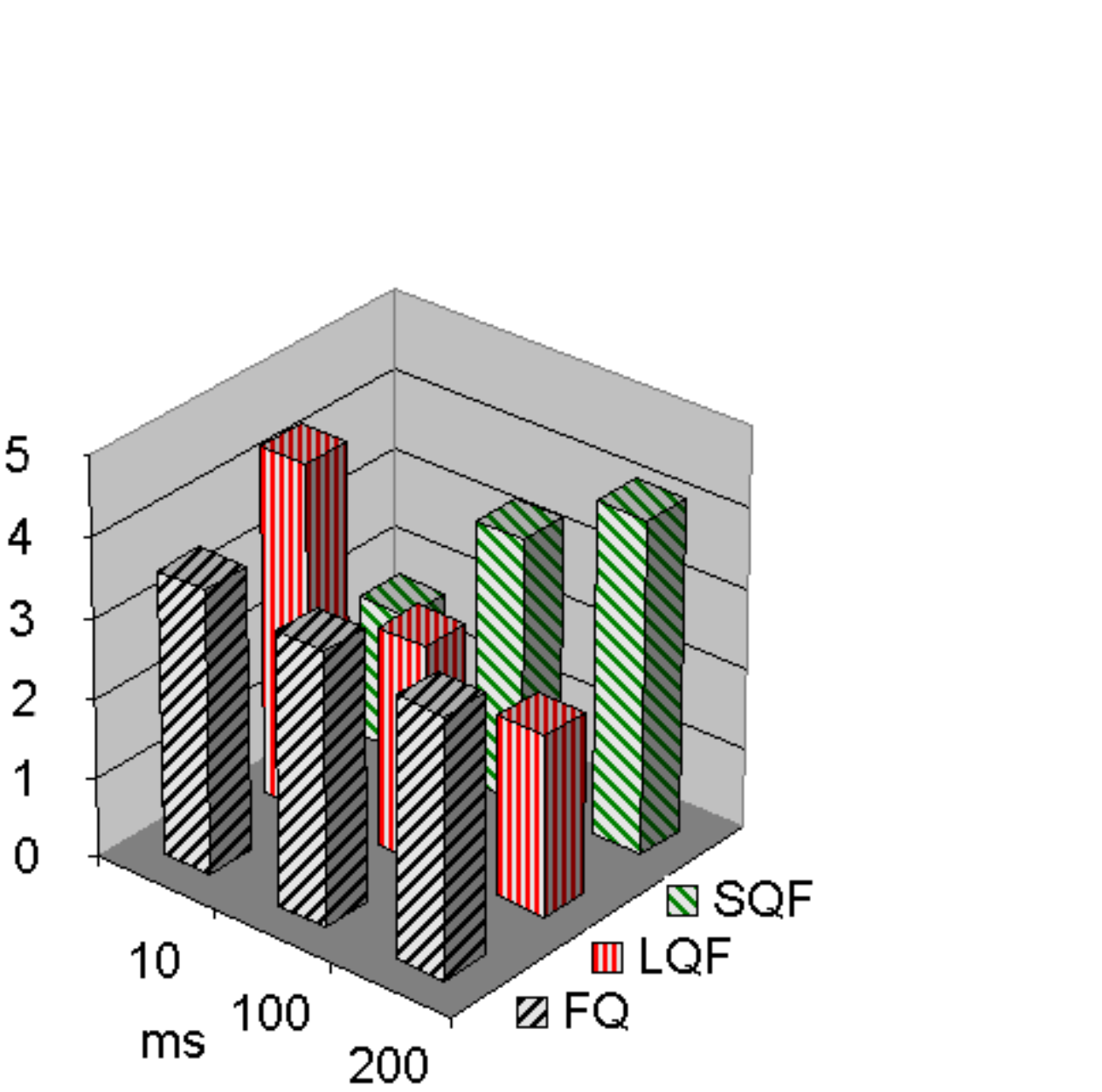}
\end{array}$
\caption{On the left: the experimental scenario. On the right: throughput's allocations.}
\label{fig:graph-test}
\end{figure}
Another interesting metric is the ratio between the received rate (throughput) and
the sending rate $R_{rcv}/R_{snd}$ which is equal to 99\% for FQ,
and  LQF and 95\% for SQF. In fact, SQF induces more losses w.r.t. the other two schedulers.
Link utilization is not affected by the scheduler employed and is always about 100\%.
All these phenomena are difficult to study and explain properly through a testbed and should
preferably be analyzed through an explanatory model.
A number of questions should find an answer through the analysis:
\begin{itemize}
\item[i)] what is the instantaneous sending rate/throughput of TCP under these three schedulers?
\item[ii)] what is the long term throughput?
\item[iii)] in a general network traffic scenario, including also UDP flows, what
is the performance of the whole system?
\end{itemize}
\begin{table}
\centering
\begin{tabular}{|cc|cc|cc|}
\hline
 \multicolumn{2}{|c}{{FQ}} & \multicolumn{2}{|c}{{LQF}} &\multicolumn{2}{|c|}{{SQF}} \\
\hline
\hline
ST & LT  & ST  & LT & ST & LT \\
0.999 & 0.999 & 0.734  & 0.736 & 0.55 & 0.735 \\
\hline
\end{tabular}
\caption{Jain index of fairness on short (0.5s) and long term (5min).}
\label{tab:fairness}
\end{table}

\section{Fluid Model}\label{sec:model}
The network scenario under study is a single bottleneck link of capacity $C$ shared by a finite number $N$ of
long-lived traffic sources of rate controlled (TCP) and rate uncontrolled (UDP) kind.
We first consider the case of rate controlled TCP sources only. The mixed TCP/UDP scenario is studied
in sec.\ref{sec:UDP}.

\subsection{Assumptions}\label{subsec:assumptions}
Each traffic source is modeled at the flow timescale through a  deterministic fluid model where the discrete packet
representation is replaced by a continuous one, either in space and in time.
Previous examples of fluid models of TCP/UDP traffic can be found in \cite{misra,ajmone,baccelli}.
Let us summarize the assumptions behind the model and give the notation (reported in Tab.\ref{tab:notation}).
\begin{itemize}
\item  A rate controlled source is modeled by a TCP Reno flows in \textit{congestion avoidance phase}, driven
by the AIMD (Additive Increase Multiplicative Decrease) rule, thus neglecting initial slow start phase and fast recovery.
\item  The \textit{round-trip delay}, $R_k(t)$ of TCP flow $k$, with  $k=1,...,N$, is defined as the sum of two terms:
$$\label{eqn:RTT} R_k(t)= PD_k+\frac{Q_k(t)}{C},$$
where $PD_k$ is the constant round-trip propagation delay, (which includes the transmission delay) and $Q_k(t)/C$ is the queuing delay,
with $Q_k(t)$ denoting the instantaneous virtual queue associated to flow $k$.
Remark that in presence of a `per-flow' scheduler  each flow $k$ is mapped into a virtual queue $Q_k$. The total queue occupancy,
denoted as $Q(t)=\sum_k Q_k(t)$, is limited to $B$.
\item \textit{Sending rate.}  The
instantaneous sending rate of flow $k$, $k=1,2,\dots,N$, denoted by $A_k(t)$, is assumed proportional to the congestion window (in virtue of Little's law), $$A_k(t)=W_k(t)/R_k(t).$$
\item \textit{Buffer management mechanism.} Under the assumption of \textit{ longest queue drop}
(LQD) as buffer management mechanism, whenever the total queue saturates, the TCP flow with the longest
virtual queue is affected by  packet losses and consequent window halvings.
\end{itemize}
\subsection{Source equations}\label{subsec:sources}
According to the usual fluid deterministic representation
(\cite{ajmone,baccelli,carofiglioSplit}),
the sending rate linearly increase as $1/R^2_k(t)$ in absence of packet losses.
Such representation usually employed for FIFO schedulers has to be modified in presence of per-flow schedulers when flows are not all simultaneously in service.
More precisely, it is reasonable to assume that flows do not increase their sending rate when not in service. It follows that in absence of packet losses, the increase
of the sending rate is given by $$\frac{1}{R^2_k(t)} \left(\I{Q(t)=0}+\frac{D_k(t)}{C}\I{Q(t)>0}\right).$$
The last expression accounts for linear increase whenever the total queue is empty and also for the reduction of the increase factor when multiple flows are simultaneously in service, proportional to their own departure rate,
$D_k(t)$.\\
Whenever a congestion event takes place (i.e. when the total queue $Q$ reaches saturation, $Q(t)=B$) and the virtual queue $Q_k(t)$ is the largest one, flow $k$ starts loosing packets in the queue at a rate proportional to the exceeding input rate at the queue, that is $(A_k(t)-C)^+$.
In addition, it halves its sending rate $A_k(t)$ proportionally to $(A_k(t)-C)^+$ (we use the convention $(\cdot)^+=\max(\cdot,0)$).
Therefore, the instantaneous sending rate of flow $k$, $k=1,\dots,N$, satisfies the following ODE (ordinary differential equation):
\begin{align}\label{eqn:ODEsource}
&\frac{dA_k(t)}{dt}=\frac{1}{R^2_k(t)} \left(\I{Q_t=0}+
\frac{D_k(t)}{C}\I{Q_t>0}\right)
-\frac{A_k(t)}{2}L_k(t-R_k(t))
\end{align}
where
$D_k$  denotes the departure rate for virtual queue $k$ at time $t$ (the additive increase takes place only when the
corresponding virtual queue is in service) and
$L_k(t)$ denotes the loss rate of flow $k$ defined by
\begin{align}\label{def:L}
 &L_k(t)=
\begin{cases}
(A(t)-C)^+\I{Q(t)=B}\I{k=\arg\max_{j} Q_j(t)} &\text{ \hspace{-6mm}\small{if }\scriptsize{$Q_j(t)<Q_k(t)$, $\forall j\neq k$}}\\
(A_k(t)-D_k(t))^+\I{Q(t)=B}\I{k=\arg\max_{j} Q_j(t)}&\text{\small{otherwise.}}
\end{cases}
\end{align}
The loss rate is proportional to the fraction of the total arrival rate $A(t)=\sum_k A_k(t)$ that
exceeds link capacity when there exists only one longest queue. In presence of multiple longest queues, the allocation of losses among flows is made according to the difference between the input and the output rate of each flow.
Of course, the reaction of TCP to the losses is delayed according to the round trip time $R_k(t)$.
\begin{obs}
Note that the assumption of zero rate increase for non-in-service flows is a consequence of the fact that the acknowledgement's rate is null in such phase. On the contrary, for FIFO schedulers, all flows are likely to loose packets and adjust their sending rate simultaneously, so  one can reasonably argue that the acknowledgment rate is never zero.
\end{obs}
\begin{table}[tb]
\begin{footnotesize}
\centering
\begin{tabular}{|l|l|}
\hline
 $C$ & Link capacity   \\
 $N$ & Number of TCP flows  \\
 $R_k(t)$ & Round trip delay of flow $k$, $k=1,...,N$ \\
 $Q_k(t)$ & Virtual queue of flow $k$, $k=1,...,N$ \\
 $Q(t)$ & Total queue of finite size $B$ \\
 $A_k(t)$ & Sending rate of flow $k$, $k=1,...,N$\\
 $D_k(t)$ & Departure rate (or throughput) of flow $k$, $k=1,...,N$\\
 $L_k(t)$ & Loss rate of flow $k$, $k=1,...,N$\\
 $\Q^{MAX}$ & Set of longest queues ($\Q^{MIN}$ similarly defined)\\
 $\A_t^B$  & Set of bottlenecked flows, i.e. $A_k(t)\ge C/N$, $k=1,...,N$ \\
\hline
\end{tabular}
\caption{Notation}\label{tab:notation}
\end{footnotesize}
\end{table}
\subsection{Queue disciplines}\label{subsec:queues}
The instantaneous occupation of virtual queue $k$, $k=1,\dots,N$, obeys to the fluid ODE:
\begin{align}\label{eqn:ODEqueuek}
&\dfrac{dQ_k(t)}{dt} =A_k(t) - D_k(t) - L_k(t).
\end{align}
In this paper we consider three different work-conserving service disciplines (for which $\sum_k D_k = C \I{Q(t) >0}$):
FQ (Fair Queuing), LQF (Longest Queue First), SQF (Shortest Queue First).
The departure rate $D_k(t)$ varies according to the chosen service discipline:
\begin{itemize}
\item FQ:
$$D_k(t)=
 \begin{cases}
\frac{C-\sum_{j \notin \A_t^B} A_j(t)} { |\A_t^B|} & \text{if $A_k(t)\in
 \A_t^B$}\\
A_k(t)& \text{if $A_k(t)\notin \A_t^B$}
\end{cases}$$
\item LQF: $$D_k(t)=C \frac{A_k(t)}{\sum_{j \in \Q^{MAX}}A_j(t)}\I{ Q_k(t)=\max_j Q_j(t)}$$
\item SQF:$$D_k(t)=C  \frac{A_k(t)}{\sum_{j \in \Q^{MIN}}A_j(t)}\I{ Q_k(t)=\min_j Q_j(t)}$$
\end{itemize}
The total loss rate is denoted by $L(t)=\sum_k L_k(t)$.	
The instantaneous occupation of the total queue $Q(t)$ is, hence, given by
\begin{align}\label{eqn:ODEqueue}
&\dfrac{dQ(t)}{dt} =A(t) - C\I{Q(t)>0} - L(t).
\end{align}
\section{Model accuracy}\label{sec:accuracy}
Before solving the model presented in Sec.\ref{sec:model}, we present some
packet level simulations using $ns2$ (\hspace{-0.2mm}\cite{ns2}) to show the accuracy of the
model.
We have implemented LQF and SQF in addition to FQ which is already available in $ns2$;
all implemented schedulers use shared memory and LQD buffer management.
Network simulations allow to monitor some variables more precisely than in a test-bed
as TCP congestion window (cwnd), and virtual queues time evolutions.
ending rate evolution is then evaluated as the ratio cwnd/RTT and queue evolution is
measured at every packet arrival, departure and drop.
This section includes some samples of the large number of simulations
run to assess model accuracy.
We present a simple scenario than counts two TCP flows with $RTTs=2ms$, $6ms$
sharing the same bottleneck of capacity $C=10$Mbps, with a line card
using a memory of size $150$kB. $ns2$ simulates IP packets of fixed MTU size
equal to 1500B.
In Figg. \ref{fig:sqf},\ref{fig:fq-lqf} we compare the system evolution
predicted by (\ref{eqn:ODEsource})-(\ref{eqn:ODEqueue}) against
queue and rate evolution estimated in $ns2$ as the ratio cwnd/RTT (congestion window over round trip time, variable in time).
Besides the intrinsic and known limitations of fluid models,  not able to
capture the burstiness at packet-level (visible in the sudden changes of queue occupancy),
the match between packet level simulations
and model prediction is remarkable.
The short timescale oscillations
observed in SQF will be better explained in Sec.\ref{sec:Analytical-results}.
\begin{figure*}[!htb]
\centering
\hspace{-2mm}\includegraphics[width=0.8\textwidth]{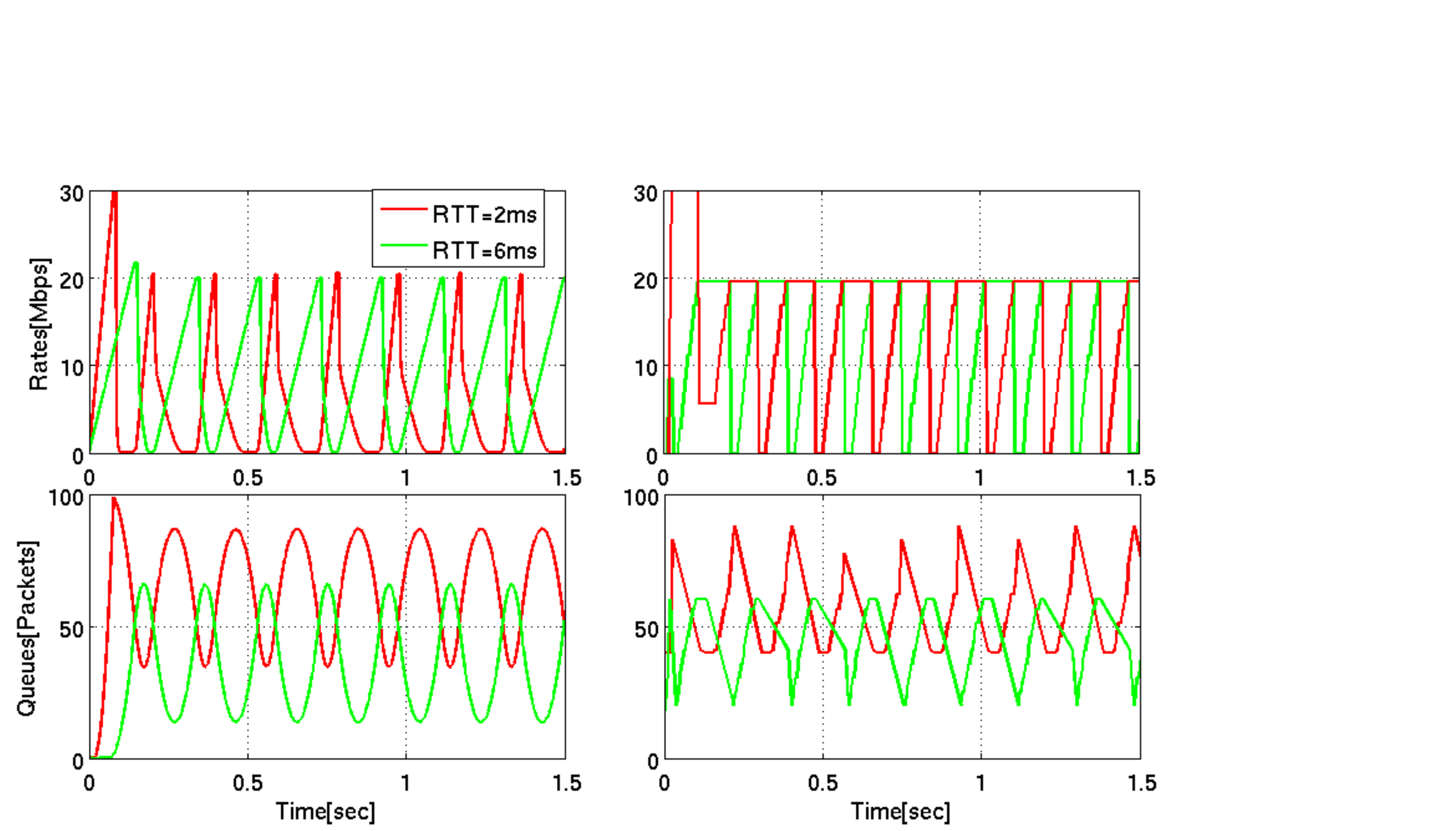}
\caption{Time evolution of rates (top) and queues (bottom) under SQF: the model on the left,  $ns2$ on the right.}
\label{fig:sqf}
\end{figure*}
\begin{figure*}[!htb]
\includegraphics[width=0.8\textwidth]{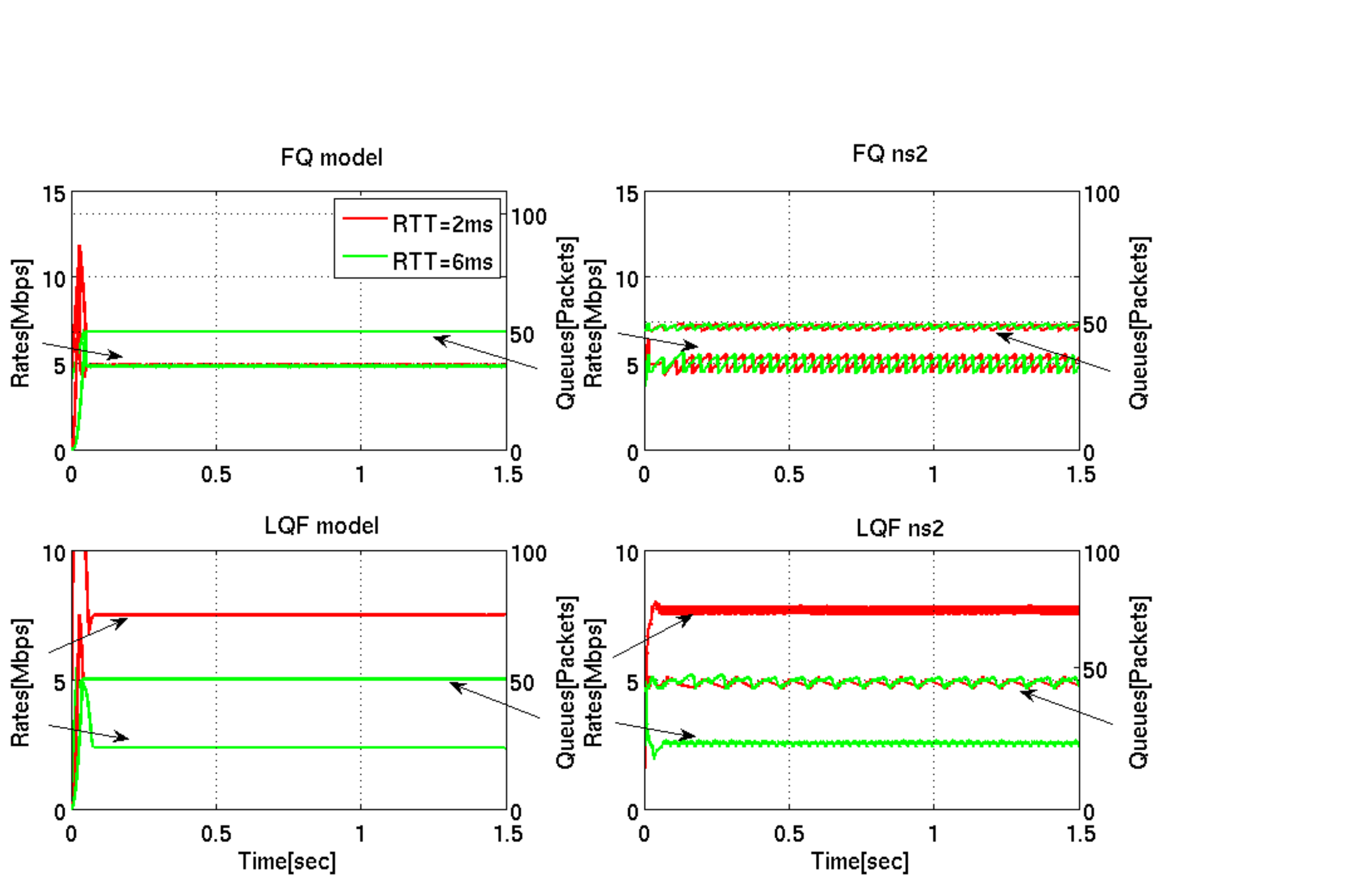}
\centering
\caption{Time evolution of rates and queues under FQ (top) and LQF (bottom): the model on the left,  $ns2$ on the right.}
\label{fig:fq-lqf}
\end{figure*}
\section{Analytical results}\label{sec:Analytical-results}
\vspace{-1mm}
Let us focus on the system of ODEs  (\ref{eqn:ODEsource})-(\ref{eqn:ODEqueue}) under the simplifying
assumption of instantaneous congestion detection
(the same assumption as in \cite{baccelli},\cite{carofiglioSplit}) and of constant round trip delay,
\begin{align}\label{eq:R}
R_k(t) \approx PD_k, k=1,\dots,N.
\end{align}
The last assumption is reasonable when the propagation delay term is predominant.
In addition, the numerical solution of (\ref{eqn:ODEsource})-(\ref{eqn:ODEqueue}) and
$ns2$ simulations confirm that the system behavior in presence of variable $R_k(t)$ and
delayed congestion detection appears to be not significantly different.
Eq.(\ref{eqn:ODEsource}) becomes
\begin{align}\label{eqn:ODEsource2}
&\frac{dA_k(t)}{dt} =\frac{1}{R^2_k} \left(\I{Q_t=0}+\frac{D_k(t)}{C}\I{Q_t>0}\right)-\frac{A_k(t)}{2}L_k(t).
\end{align}
In the following, we analytically characterize the solution of the system of ODEs
(\ref{eqn:ODEqueue})-(\ref{eqn:ODEsource2}) in presence of $N=2$ flows and
for the three scheduling disciplines under study. We first focus on SQF scheduling
discipline, before
studying the more intuitive behavior of the system under LQF and FQ in Sec.\ref{sec:LQF}, Sec.\ref{sec:FQ}.
\subsection{Shortest Queue First (SQF) scheduling discipline}\label{sec:SQF}
For the ease of exposition, denote $\alpha=\frac{1}{R_1^2}$ and $\beta=\frac{1}{R_2^2}$ and take $\alpha > \beta$ (the same arguments hold in the dual case $\alpha \le \beta$ where replacing $\alpha$ with $\beta$ and viceversa).
\begin{lem}
 The dynamical system described by (\ref{eqn:ODEqueue})-(\ref{eqn:ODEsource2}) under SQF scheduling discipline admits as unique stationary solution the limit-cycle composed by:
\begin{itemize}
 \item \textit{phase $A_2^{ON}$ }:
 $\forall t \in$ phase $A_2^{ON}$ (when  the origin coincides with the beginning of the phase)
\begin{align*}
&\widetilde{A}_1(t)= 2\sqrt{\beta}2C\widetilde{f}(2C,0,t), \qquad\widetilde{A}_2(t)=\beta t, \nonumber \\
&\widetilde{Q}_2(t)=\frac{B}{2}+\int_{0}^t (\beta u -C)du,\quad\widetilde{Q}_1(t)=B-Q_2(t),
\end{align*}
where $\widetilde{f}(\cdot)$ denotes the limiting composition of $f(\cdot)$ functions,
$\widetilde{f}(\cdot) \triangleq f\circ f \circ \dots \circ f
$ and
\begin{align}\label{def:f}
&f(a,b,t)\triangleq e^{-\frac{t}{2}(2b-2C+\beta t)}/ \\
&\left(2\sqrt{\beta}+a e^{\frac{(b-C)^2}{2\beta}}\sqrt{2\pi}\left(Erf\left(\frac{b-C+\beta t}{\sqrt{2\beta}}\right)-Erf\left(\frac{b-C}{\sqrt{2\beta}}\right)\right)\right) \nonumber
\end{align}
 \item \textit{phase $A_1^{ON}$}:
 $\forall t \in$ phase $A_1^{ON}$ (when  the origin coincides with the beginning of the phase)
\begin{align*}
&\widetilde{A}_1(t)= \alpha t, \qquad\widetilde{A}_2(t)=2\sqrt{\alpha}2C \widetilde{g}(2C,0,t), \nonumber \\
&\widetilde{Q}_1(t)=\frac{B}{2}+\int_{0}^t (\alpha u -C)du, \quad\widetilde{Q}_2(t)=B-Q_1(t),
\end{align*}
where $\widetilde{g}(\cdot)$ is defined by symmetry w.r.t. $\widetilde{f}(\cdot)$ and
\begin{align}\label{def:g}
&g(a,b,t)\triangleq e^{-\frac{t}{2}(2b-2C+\alpha t)}/  \\&\left(2\sqrt{\alpha}+a e^{\frac{(b-C)^2}{2\alpha}}\sqrt{2\pi}\left(Erf\left(\frac{b-C+\alpha t}{\sqrt{2\alpha}}\right)-Erf\left(\frac{b-C}{\sqrt{2\alpha}}\right)\right)\right) \nonumber
\end{align}
\end{itemize}
The duration of phase $A_1^{ON}$ is $2C/\alpha$, while that of phase $A_2^{ON}$ is $2C/\beta$.
The period to the limit cycle is, therefore, equal to\\
$\qquad\qquad T=2C\left( \frac{1}{\alpha}+\frac{1}{\beta}\right).$
\end{lem}
\begin{proof}
The proof is structured into three steps: \textbf{1)} we  study the transient regime
and analyze how the system reaches the first \textit{saturation point},
i.e. the time instant denoted by $t_s$ at which total queue occupancy $Q$ attains saturation ($Q=B$),
$$\label{def:t_sat}
 t_s \triangleq  \inf\left \{ t>0, Q(t)=B \right \}.
$$
\textbf{2)} We show that once reached this state, the system does not leave the \textit{saturation regime}, where the buffer remains full.
\textbf{3)} Under the condition $Q(t)=B, \quad  \forall t>t_s$, we characterize the cyclic evolution of virtual queues and rate and eventually
show the existence of a unique limit-cycle in steady state whose characterization is provided analytically.\\
\textbf{1)} With no loss of generality, assume  the system starts with $Q$ empty at $t=0$, i.e. $Q_1(0)=Q_2(0)=0$ and initial rates, $A_1(0)=A_1^0$,$A_2(0)=A_2^0$.
At $t=0$, TCP rates, $A_1$, $A_2$ start increasing  according to (\ref{eqn:ODEsource2}) and the virtual queues are simultaneously served at rate $A_1(t)$, $A_2(t)$ until it holds $A(t)<C$.
Hence, the queue $Q$ remains empty until the input rate to the buffer, \\
$$A(t)=A_1(t)+A_2(t)=A_1^0+A_2^0+(\alpha+\beta)t$$
exceeds the link capacity $C$, in $t_0=(C-A_1^0-A_2^0)/(\alpha+\beta)$.
From this point on, the total queue starts filling in and TCP rates keep increase in time (with $A(t)>C$, $\forall t>t_0$.
The system faces no packet losses until the first saturation point.
Departing from a non-empty queue clearly accelerates the attainment of buffer saturation, but it is always true that
at $t_s$, the buffer is full ($Q(t_s)=B$) and $A_1(t_s)+A_2(t_s)>C)$.
In general one can state that independently of the state of the system at $t=0$, the rate increase in absence of packet losses guarantees that the system reaches in a finite time  the total queue saturation. When this happens, at $t=t_s$, the total input rate is greater than the output link capacity, i.e. $A(t_s)>C$, which is equivalent to say that $dQ_t/dt|_{t=t_s}>0$.
\\\textbf{2)} Once reached the saturation, the auxiliary result in appendix \ref{app:1} proves that the system remains in the \textit{saturation regime}, that is $Q(t)=B$, $\forall t>t_s$.
\begin{obs}
 One can argue that in practice the condition of full buffer is only verified within limited time intervals, and that the queue occupancy goes periodically under $B$, when the buffer is emptying. Actually, $ns2$ simulations and experimental tests confirm that the queue occupancy can instantaneously be smaller than $B$, but overall the  average queue occupancy is approximately equal to $B$. Indeed,
the difference w.r.t. a Drop Tail/FIFO queue is that SQF with LQD  mechanism eliminates the loss synchronization induced by drop tail under FIFO scheduling discipline, so that in our setting when a flow experiments packet losses, the other one still increases and feeds the buffer.
\end{obs}
\textbf{3)}
Suppose $Q_1$ is the longest queue at $t_s$ (the dual case is completely symmetric).
The rate $A_1$ sees a multiplicative decrease according to (\ref{eqn:ODEsource2})
as $Q(t_s)=B$ and $Q_1$ is the longest virtual queue. By solving (\ref{eqn:ODEsource2}), the resulting rate decay of $A_1$ is $\forall t>t_s$
\begin{align}\label{A1*_after}
A_1(t)=2\sqrt{\beta}A_1(t_s)f(A_1(t_s),A_2(t_s),t-t_s),
\end{align}
where $f(\cdot)$ is defined in (\ref{def:f}).
At the same time, $A_2$ keeps increasing linearly.
In terms of queue occupation, since we proved that $Q(t)=B$ from $t_s$ on, it means that the decrease
due to packet drops in $Q_1$  is compensated by
the increase of $Q_2$.
In fact, 
no matter the values of $Q_1$ and $Q_2$ in $t_s$, the rate at which $Q_1$ decreases is a function of the rate at which $Q_2$ increases.
Indeed, $\forall t>t_s$,
\begin{align*}
 Q_1(t)&= Q_1(t_s)+\int_{t_s}^{t} A_1(v)dv-\int_{t_s}^{t} (A_1(v)+A_2(v)-C)dv \nonumber\\
\vspace{-1mm}  &=B-Q_2(t_s)-\int_{t_s}^{t} (A_2(v)-C)dv =B- Q_2(t).
\end{align*}
It follows that the two queues become equal when $Q_1=Q_2=\frac{B}{2}$.
We can denote by $t_1=t_s+\tau_1$ the first \textit{queue meeting point},
\begin{align}\label{def:t1}
t_1 &\triangleq  \inf\left \{ t>t_s, \int_{t_s}^{t} (A_2(v)-C) dv=\frac{B}{2} \right \}.
\end{align}
The state of the system in $t_1$ is,
\begin{align}\label{ic_t1}
A_1(t_1)= \epsilon_1, &~ A_2(t_1)= \beta t_1, & Q_1(t_1)=Q_2(t_1)=\frac{B}{2}.
\end{align}
where we defined $\epsilon_1>0$, the value $A_1(t_1)$ given by (\ref{A1*_after}), which we recall is a function of $A_1(t_s)$.
From $t_1$ on, the two flows and thus the two virtual queues invert their roles: $Q_1$ enters in service, whereas $Q_2$ suffers from packet losses and, consequently, $A_2$ observes a multiplicative rate decrease, while $A_1$ increases linearly.
Since $L(t)>0$, $\forall t>t_1$,
\begin{align*}
Q_1(t)&= Q_1(t_1)+\int_{t_1}^{t} (A_1(u)-C) \\
&=\frac{B}{2}+\int_{t_1}^{t} (\epsilon_1 +\alpha (v-t_1)-C) dv=B- Q_2(t). \nonumber
\end{align*}
and $Q_1(t'_1)=Q_2(t'_1)=B/2$, if we denote by $t'_1=t_1+\tau'_1$ the next queue meeting point,
\begin{align}\label{def:t1primo}
t'_1\triangleq &\inf\left\{t>t_1, \frac{B}{2}+\int_{t_1}^{t'_1} \left( \epsilon_1+\alpha(v-t_1)-C \right) du=\frac{B}{2} \right \} 
\end{align}
The state of the system in $t'_1$ is,
\begin{align*}
&A_1(t'_1)=\epsilon_1+\alpha \tau'_1,&A_2(t'_1)= \gamma_1, &~ Q_1(t'_1)=Q_2(t'_1)=\frac{B}{2}.
\end{align*}
where $\gamma_1=A_2(t'_1)=2\sqrt{\alpha}A_2(t_1)g(A_2(t_1),A_1(t_1),t'_1-t_1)$ ($\gamma_1$ is a function of $A_2(t_s)$) and
$g(\cdot)$ is defined in (\ref{def:g}).
\begin{figure}[htb]
\begin{center}
\includegraphics[width=0.8\textwidth,height=0.35\textwidth]{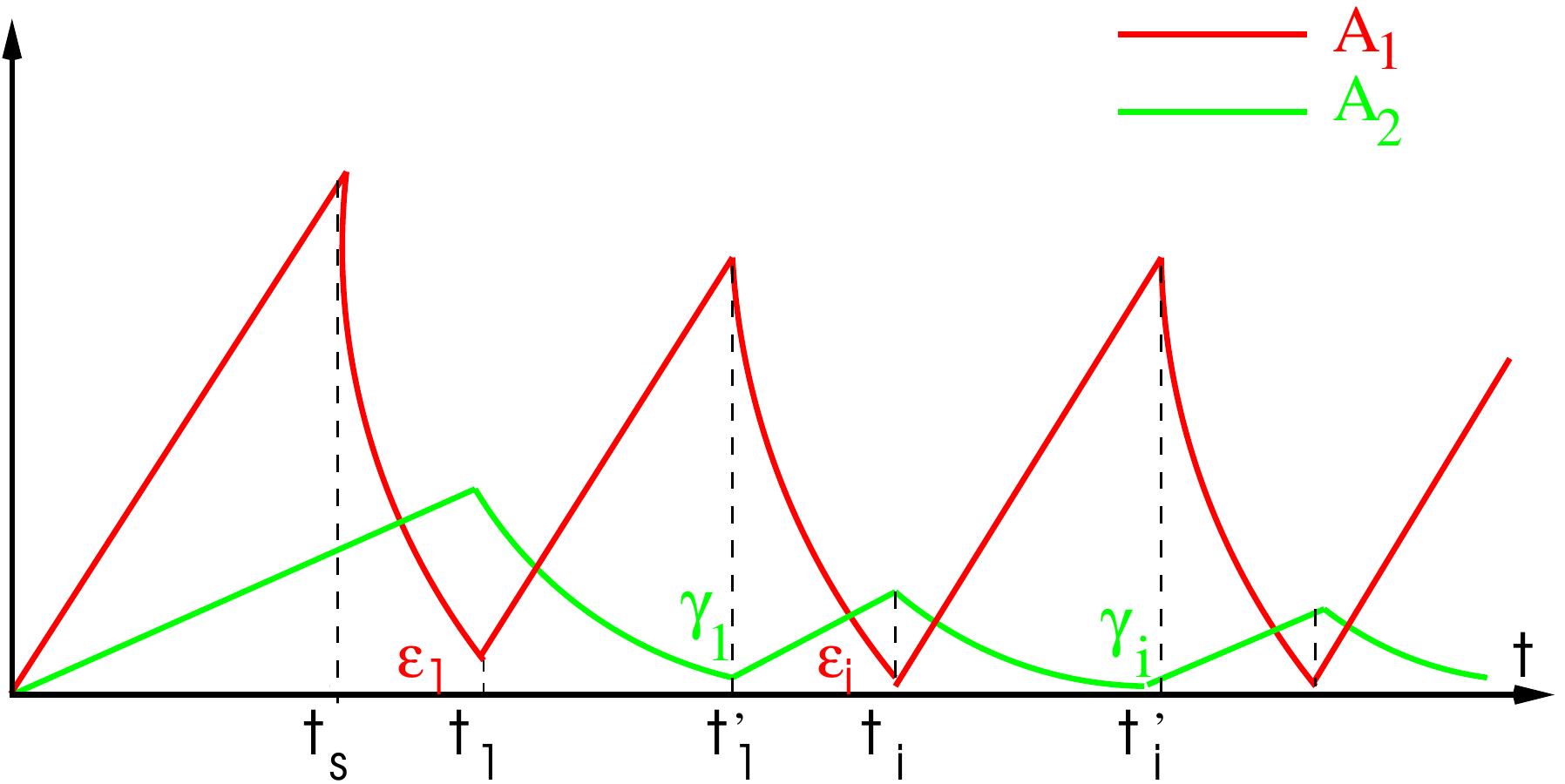}
\end{center}
\caption{Cyclic time evolution of TCP rates under SQF scheduling.}
 \label{fig:rates-cycle}
\end{figure}
\\Iteratively, one can construct a cycle for the virtual queue $Q_1$ ($Q_2$)  composed by two phases:
\begin{itemize}
 \item \textit{phase $A_2^{ON}$ }: when $Q_1$ ($Q_2$) goes from $B/2$ to $B/2$ remaining smaller (greater) than $B/2$ at a rate $C-A_2$ (respectively $A_2-C$);
 \item \textit{phase $A_1^{ON}$}:  when $Q_1$ ($Q_2$) goes from $B/2$ to $B/2$ remaining greater (smaller) than $B/2$ at a rate $C-A_1$ (respectively $A_1-C$).
\end{itemize}
The duration of these phases depends on the values of $A_2$ and $A_1$ respectively at the beginning of phase $A_1^{ON}$ and phase $A_2^{ON}$.
As in Fig.\ref{fig:rates-cycle}, denote by $\epsilon_i$ and $\gamma_i$ the initial values of $A_1$ and $A_2$ at the beginning of the $i^{th}$ phase $A_1^{ON}$ and $A_2^{ON}$ respectively.
Thanks to the auxiliary result $2$ in appendix \ref{app:2}, we prove that in steady state such initial values are zero and easily conclude the proof.
Let us now provide analytical expressions for the mean stationary values of sending rate and throughput. As a consequence, one can also compute the mean values of the stationary queues,
$\widetilde{Q}_1$, $\widetilde{Q}_2$ (one smaller, the other greater than $B/2$).
\end{proof}
\begin{cor}\label{cor:SQF}
The \textit{mean sending rates} in steady state are given by:
$$ \overline{A}_1\approx\left(\frac{\beta}{\alpha+\beta}C+\frac{\alpha\beta}{2C(\alpha+\beta)}\log\left(1+\frac{2\sqrt{2\pi}}{\sqrt{\beta}}C e^{\frac{C^2}{2\beta}}Erf\left(\frac{C}{\sqrt{2\beta}}\right)\right)\right)$$
$$\overline{A}_2\approx\left(\frac{\alpha}{\alpha+\beta}C+\frac{\alpha\beta}{2C(\alpha+\beta)}\log\left(\hspace{-1mm}1+\frac{2\sqrt{2\pi}}{\sqrt{\alpha}}C e^{\frac{C^2}{2\alpha}}Erf\left(\frac{C}{\sqrt{2\alpha}}\right)\right)\right)$$
The \textit{mean throughput values} in steady state are given by:
$$ \overline{X}_1=\frac{\beta}{\alpha+\beta}C,\quad \overline{X}_2=\frac{\alpha}{\alpha+\beta}C.$$
\end{cor}
The mean virtual queues in steady state are given by:
$$ \overline{Q}_1=\frac{B}{2}+\frac{C^2(\alpha-\beta)}{3\alpha\beta}C,\quad \overline{Q}_2=\frac{B}{2}-\frac{C^2(\alpha-\beta)}{3\alpha\beta}C.$$
\begin{proof}
The proof is reported in appendix \ref{app:3}.
\end{proof}
\subsection{Longest Queue First (LQF) scheduling discipline}\label{sec:LQF}
\begin{lem}
 The dynamical system described by (\ref{eqn:ODEqueue})-(\ref{eqn:ODEsource2}) under LQF scheduling discipline admits as unique stationary solution,
\begin{align*}
&\overline{A}_1\approx\frac{\alpha}{\alpha+\beta}C, &\overline{A}_2\approx\frac{\beta}{\alpha+\beta}C, \nonumber \\
&\overline{X}_1=\frac{\alpha}{\alpha+\beta}C, &\overline{X}_2=\frac{\beta}{\alpha+\beta}C, \nonumber \\
&\overline{Q}_1=\overline{Q}_2=\frac{B}{2}.
\end{align*}
\end{lem}
\begin{proof}
As in the case of SQF, no matter the initial condition the system reaches the first saturation point,
$ t_s \triangleq  \inf\left \{ t>0, Q(t)=B \right \}$
with $A(t_s)>C$ and the buffer enters the \textit{saturation regime}, that is $Q(t)=B$, $\forall t>t_s$.
As an example, we can consider the case of initial condition $A_1^0,A_2^0$, $Q_1^0=Q_2^0=0$.
At $t=0$, $A_1,A_2$ start increasing as for SQF while $Q$ remains empty
until the input rate to the buffer, $$A(t)=A_1^0+A_2^0+(\alpha+\beta)t$$
exceeds the link capacity $C$, in $t_0=(C-A_1^0-A_2^0)/(\alpha+\beta)$.
From this point on, the total queue starts filling in and TCP rates keep increasing in time according to (\ref{eqn:ODEsource2}) with the only difference, w.r.t. the case of SQF, that
$Q_1>Q_2$ is the longest queue served until buffer saturation and $Q_2$ only gets the remaining service capacity, when there is one.
The system faces no packet losses and $Q_1$ remains in service until  $t_s$.\\
At $t_s$,
suppose with no loss of generality (the other case is symmetrical) that $Q_1$ is the longest virtual queue. The rate of $TCP1$, $A_1$ sees a multiplicative decrease according to (\ref{eqn:ODEsource2}).
By solving (\ref{eqn:ODEsource2}), the resulting rate decay of $A_1$ is $\forall t>t_s$
\begin{align}\label{A1*_after_LQF}
&A_1(t)=2\sqrt{\beta}A_1(t_s)f(A_1(t_s),A_2(t_s),t-t_s),
\end{align}
where $f(\cdot)$ is defined in (\ref{def:f}).
At the same time, $A_2$ increases linearly.
Once proved that $Q(t)=B$, $\forall t>t_s$ (see Appendix \ref{app:1}), no matter the values of $Q_1$ and $Q_2$ in $t_s$, it holds
 $Q_1(t)
        =B- Q_2(t).$
Thus, the two queues meet when $Q_1=Q_2=\frac{B}{2}$, at
what we can denote by $t_1=t_s+\tau_1$, the first \textit{queue meeting point}.
Once the queues reach the state $Q_1=Q_2=B/2$, they cannot exit this state. Indeed, if we rewrite (\ref{eqn:ODEqueuek}), it gives:
\begin{align}\label{eq:fixedpoint_Q_LQF}
&\frac{dQ_k(t)}{dt}=A_k(t)-\frac{A_k(t)}{A(t)}C-\frac{A_k(t)}{A(t)}(A(t)-C)=0
\end{align}
with $k=1,2$ independently of the rate values.
It implies that,
$\label{mean_queues_LQF}\overline{Q}_1=\overline{Q}_2=\frac{B}{2}.$
This can be explained by the following consideration: in presence of two longest queues $Q_1(t)=Q_2(t)=B/2$, both queues are served proportionally
to the input rates $A_1(t)$ and $A_2(t)$, and loose packets in the same proportion, so preventing each queue to become greater (smaller) than the other.\\
About rates evolution, we now have two ODEs not anymore coupled with the queue evolution,
\begin{align}\label{eq:fixedpoint_A_LQF}
&\frac{dA_k(t)}{dt}=\frac{1}{R_k^2}\frac{A_k(t)}{A(t)}-\frac{1}{2}A_k\frac{A_k(t)}{A(t)}(A(t)-C)=0 \quad k=1,2
\end{align}
where we recall that we proved in appendix \ref{app:1} that $A(t)>C$, $\forall t>t_s$.
The system of ODEs can be easily solved and it admits one stationary solution:
\begin{align}\label{mean_rates_LQF}
\overline{A}_1&=\frac{\alpha}{\alpha+\beta}\frac{C}{2}\left(1+\sqrt{1+\frac{8(\alpha+\beta)}{C^2}}\right);
~\overline{A}_2=\frac{\beta}{\alpha}\overline{A}_1
\end{align}
In order to compute the mean throughput values in steady state, it suffices to observe
that $\overline{X}_k=C(\overline{A}_k/\overline{A})$, $k=1,2$,
which concludes the proof. Note that in practice, $8(\alpha+\beta)<<C^2$, that implies $\overline{A}_k\approx\overline{X}_k$, $k=1,2$.
\end{proof}
\subsection{Fair Queuing (FQ) scheduling discipline}\label{sec:FQ}
\begin{lem}
 The dynamical system described by (\ref{eqn:ODEqueue})-(\ref{eqn:ODEsource2}) under FQ scheduling discipline admits a unique stationary solution,
\begin{align*}
\overline{A}_1&=\frac{C}{4}\left(1+\sqrt{1+\frac{4\alpha}{C^2}}\right);
~~\overline{A}_2=\frac{C}{4}\left(1+\sqrt{1+\frac{4\beta}{C^2}}\right); \nonumber \\
\overline{X}_1&=\overline{X}_2=\frac{C}{2};
~~\overline{Q}_1=\overline{Q}_2=\frac{B}{2}.
\end{align*}
\end{lem}
\begin{proof}
As for SQF or LQF, no matter the initial condition the system reaches the first saturation point,
$t_s \triangleq  \inf\left \{ t>0, Q(t)=B \right \}$
with $A(t_s)>C$
and the buffer enters  the \textit{saturation regime}, that is $Q(t)=B$, $\forall t>t_s$.
The evolution of the system until $t_s$ differs from that of $SQF$ or $LQF$,
as the two flows are allocated the fair rate until the virtual queues start filling in (when $A(t)$ exceeds $C$) and then they are served at capacity $\frac{C}{2}$ independently of $\alpha$ and $\beta$.
At $t_s$, the system suffers from packet losses at rate $L(t)=A(t)-C$ and the virtual queues $Q_1$, $Q_2$ evolve towards the state
$Q_1(t)=Q_2(t)=B/2$, 
\begin{align*}
 Q_1(t)=& Q_1(t_s)+\int_{t_s}^{t} (A_1(v)-\frac{C}{2})dv-\int_{t_s}^{t} (A_1(v)+A_2(v)-C)dv \nonumber\\
        &=B-Q_2(t_s)-\int_{t_s}^{t} (A_2(v-t_s)-\frac{C}{2})dv \nonumber \\
        &=B- Q_2(t).
\end{align*}
Once the queues reach the state $Q_k=B/2$, $A_k>C/2$, $k=1,2$ (otherwise one of the virtual queue would be empty) and the queues cannot exit this state.
Indeed, if we rewrite the ODE (\ref{eqn:ODEqueuek}), it gives:
\begin{align}\label{eq:fixedpoint_Q_FQ}
&\frac{dQ_k(t)}{dt}=A_k(t)-\frac{C}{2}-\left(A_k(t)-\frac{C}{2}\right)=0
\end{align}
with $k=1,2$ independently of the rate values.
It implies that,
\begin{align}\label{mean_queues_FQ}
\overline{Q}_1=\overline{Q}_2=\frac{B}{2}.
\end{align}
About the rates evolution, we now have two ODEs not anymore coupled with the queue evolution,
\begin{align}\label{eq:fixedpoint_A_FQ}
&\frac{dA_k(t)}{dt}=\frac{1}{2 R_k^2}-\frac{A_k(t)}{2}\left(A_k(t)-\frac{C}{2}\right)=0, ~ k=1,2
\end{align}
The system of ODEs can be easily solved and it admits one stationary solution:
\begin{align}\label{mean_rates_FQ}
&\overline{A}_1=\frac{C}{4}\left(1+\sqrt{1+\frac{4\alpha}{C^2}}\right); &\overline{A}_2=\frac{C}{4}\left(1+\sqrt{1+\frac{4\beta}{C^2}}\right)
\end{align}
In order to compute the mean throughput values in steady state, it suffices to observe that virtual queues are not empty and $A_k>C/2$, $k=1,2$,
therefore $\overline{X}_k=C/2$,
 $k=1,2$. In addition,
note that, if $4\alpha<<C^2$, condition usually verified in practice, $A_k\approx X_k$, $k=1,2$.
\end{proof}




\section{Model Extensions}\label{sec:extensions}
\subsection{Extension to $N$ flows}
In Sec.\ref{sec:Analytical-results} we studied the case of $N=2$ TCP flows and analytically characterize
the stationary regime. For the case of $N>2$ flows,
one can generalize the analytical results obtained in the two-flows scenario.
More precisely, under the SQF scheduling discipline, the resulting steady state
solution of (\ref{eqn:ODEqueue})-(\ref{eqn:ODEsource2}) is a limit-cycle
composed by $N$ phases,
each one denoted as $A_k^{ON}$ phase (where flow $k$ is in service) and of duration $\frac{2C}{\alpha_k}$, with $\alpha_k=\frac{1}{R_k^2}$.
The mean stationary throughput associated to flow $k$ is:
\begin{align}\label{SQF_N_flows_X}
 \quad SQF: \qquad \overline{X}_k=\frac{1/\alpha_k}{\sum_{j=1}^{N} 1/\alpha_j}C
\end{align}
\\Under $LQF$ scheduling discipline, one gets a generalization of (\ref{mean_rates_LQF}), where the mean stationary throughput associated to flow $k$ is:\\
\begin{align}\label{LQF_N_flows_X}
 \quad LQF: \qquad \overline{X}_k=\frac{\alpha_k}{\sum_{j=1}^N \alpha_j}C
\end{align}
\\Similary, under FQ sheduling disciplines, it can be easily verified that there is still a fair allocation of the capacity $C$ among flows:\\
\begin{align}\label{FQ_N_flows_X}
 \quad FQ: \qquad \overline{X}_k=\frac{C}{N}
\end{align}
\subsection{Extension to UDP traffic}\label{sec:UDP}
In this section we extend the analysis to the case of rate uncontrolled sources (UDP)
competing with TCP traffic. We simply model a UDP source as a CBR flow
with constant rate $X^{UDP}$.
Consider the mixed scenario with one TCP source and one UDP source traversing a single bottleneck link
of capacity $C$. The system is described by a set of ODEs:
\begin{align}\label{ODE_TCP-UDP}
&\frac{dA_1(t)}{dt} =\alpha\left(\I{Q_t=0}+
\frac{D_1(t)}{C}\I{Q_t>0}\right)-\frac{1}{2}A_1(t)L_1(t);  \\
&\frac{dA_2(t)}{dt}=0;\quad \frac{dQ_k(t)}{dt}=A_k(t)-D_k(t)-L_k(t), \quad k=1,2.
\nonumber
\end{align}
with $A_1$ the sending rate of the TCP flow, and $A_2(t)=A_2(0)=X^{UDP}\le C$ the sending rate of the UDP flow.
The definition of $L_k(t)$ is the same as in (\ref{def:L}).\\
The interesting metric to observe in the mixed TCP/UDP	 scenario is the loss rate experienced by UDP in steady state, $\overline{L}^{UDP}=\overline{L_2}$,
which allows one to evaluate the performance of streaming applications against those carried by TCP.\\
Under SQF scheduling discipline, independently of the initial condition, the system reaches a first saturation point, $t_s$, where $Q(t_s)=B$. After $t_s$, the flow with the longest queue suffers from packet losses. Depending on the initial condition, it may be the UDP flow to first loose (if $Q_2(t_s)>Q_1(t_s)$) and in this case its virtual queue decreases until $t_1$, where $Q_1(t_1)=Q_2(t_1)=B/2$ before entering in service and finally emptying.
Indeed, for $t_s<t<t_1$,
$ Q_2(t)=Q_2(t_s)+\int_{t_s}^{t} X^{UDP}-(A_1(u)+X^{UDP}-C)du = B-Q_1(t),$
and for $t>t_1$,
$ Q_2(t)=\frac{B}{2}+\int_{t_1}^t (X^{UDP}-C) du.$
Whether it is TCP to loose in $t_s$, virtual queues both decrease until $Q_2$ empties, i.e. $\forall t>t_s$,
$Q_2(t)=Q_2(t_s)+\int_{t_s}^t (X^{UDP}-C) du.$
Once emptied, $Q_2$ remains equal to zero, since it is fed at rate $X^{UDP}\leq C$.
The remaining service capacity $C-X^{UDP}$ is allocated to the TCP flow, that behaves as a TCP flow in isolation on a bottleneck link of reduced capacity, equal to $C-X^{UDP}$.
Thus, the loss rate of UDP is zero in steady state,
\begin{align}\label{Loss_rate_UDP_SQF}
 \quad SQF: \quad\overline{L}^{UDP}=0.
\end{align}
When SQF is replaced by LQF, the outcome is the opposite. In fact, once the system reaches the first buffer saturation in $t_s$, one can easily observe that virtual queues
tend to equalize, either when TCP or UDP is the flow affected by packet losses. When $Q_1(t_1)=Q_2(t_1)=B/2$ in $t=t_1$, virtual queues do not change anymore, since
$dQ_{max}(t)=A_1-C+X^{UDP}-(A_1+X^{UDP}-C)=0, \quad t>t_1$, with $Q_{max}$ equal to $Q_1$ or $Q_2$ depending on the initial condition, and the
smaller one equal to $B-Q_{max}$.
The first ODE in (\ref{ODE_TCP-UDP}) related to $A_1$ admits one stationary solution, that is,
$$\label{eq:A1-UDP}
\overline{A}_1=\frac{C-X^{UDP}}{2}\left( 1+\sqrt{1+\frac{8\alpha}{(C-X^{UDP})^2}}\right).$$
Hence, the loss rate of UDP results to be 
\begin{align}\label{Loss_rate_UDP_LQF-prem}
\overline{L}^{UDP}=C\frac{X^{UDP}}{X^{UDP}+\overline{A}_1}
\end{align}
By replacing the expression of $\overline{A}_1$ in (\ref{eq:A1-UDP}), we get
\begin{align*}
\overline{L}^{UDP}&>C\frac{X^{UDP}}{X^{UDP}+\frac{C-X^{UDP}}{2}\left(2+\sqrt{\frac{8\alpha}{(C-X^{UDP})^2}}\right)} \nonumber \\&>C\frac{X^{UDP}}{C+\sqrt{2\alpha}}
\end{align*}
and 
\begin{align*}
\overline{L}^{UDP}<&C\frac{X^{UDP}}{X^{UDP}+C-X^{UDP}} =X^{UDP}.
\end{align*}
It follows that:
\begin{align}\label{Loss_rate_UDP_LQF}
\quad LQF: \overline{L}^{UDP}\approx X^{UDP}
\end{align}
under the condition $2\alpha<<C^2$, usually verified in practice.
It implies that the UDP loss rate is almost equal to its sending rate, and so the UDP throughput
is almost zero, exactly the opposite result w.r.t. SQF.
Finally, under FQ as scheduling discipline, it is easy to observe that both flows tend to the \textit{fair rate},
$\overline{X}_k=\frac{C}{2}$, $k=1,2$
and the loss rate of UDP is, hence, given by the difference
\begin{align}\label{Loss_rate_UDP_FQ}
\quad FQ: \quad\overline{L}^{UDP}=\left(X^{UDP}-\frac{C}{2}\right)^+.
\end{align}

\section{Numerical results}\label{sec:numerics}
This section presents numerical evaluations of formulas obtained
in Sec.\ref{sec:Analytical-results} and Sec.\ref{sec:extensions}
and comparisons with packet level simulations.
We focus on TCP sending rates, TCP throughputs and UDP loss rates
for a selected scenario among a larger set.
We consider a $10$Mbps bottleneck link for three TCP flows with different
RTTs. Tab.\ref{tab:notation} reports the numerical values of the
comparison, showing a good agreement between model and $ns2$.
In $ns2$ the link is almost fully utilized, although the model predicts
100\% utilization, this mirroring the model condition $A(t)\ge C$ in steady state.
The largest deviation
is registered for the TCP sending rate in presence of SQF while
the throughput is correctly predicted. We recall that while the throughput
formula is exact, the sending rate formula is approximated.
\begin{table}
\centering
\begin{tabular}{|l|rr|rr|rr|}
\hline
& \multicolumn{2}{|c}{{FQ [Mbps]}} & \multicolumn{2}{|c}{{LQF [Mbps]}} &\multicolumn{2}{|c|}{{SQF [Mbps]}} \\
\hline
\hline
& $ns2$ & Model& $ns2$ & Model& $ns2$ & Model\\
$\overline{A}_1$  & 5.26 & 5.00 & 7.48 & 7.35 & 5.8 & 6.32 \\
$\overline{A}_2$  & 4.93 & 5.00 & 2.63 & 2.65 & 8.1 & 8.68 \\
$\overline{A}$    & 10.2 & 10.01 & 10.01 & 10.00 & 13.90 &  15.00\\
$\overline{X}_1$ & 5.12 & 5.00 & 7.37 & 7.35 & 2.90 & 2.65 \\
$\overline{X}_2$ & 4.79 & 5.00 & 2.47 & 2.65 & 7.08 & 7.35 \\
$\overline{X}$ & 9.91 & 10.00  & 9.85 & 10.00 & 9.99 & 10.00 \\
\hline
\end{tabular}
\caption{Numerical comparison between $ns2$ and the model: $C=10$Mbit/sec, $R_1=20ms$, $R_2=50ms$.}
\label{tab:numerics}
\end{table}
We also present, in Tab.\ref{tab:numerics2}, a scenario on the performance of TCP and UDP  sharing a common bottleneck
of $10$Mbps, where the stream rate of UDP $X_{UDP}$ is taken first smaller than the fair rate ($C/2$) and then
larger. Results highlight that UDP stalls under LQF, while
it gets prioritized under FQ as long as the stream rate is smaller than
the fair rate.
The possibility to implicitly differentiate traffic through FQ has been exploited in \cite{hpsr,ancs05} for backbone
and border routers where the fair rate is supposed to be very large. SQF appears to be much more effective
to fulfill this task as shown in Tab.\ref{tab:numerics2} (and through experimentation in \cite{ostallo}),
since it gives priority to UDP flows with any stream rate and allocates the rest of the bandwidth to TCP
(see Fig.\ref{fig:sqf-tcp-dup} on top, for the time evolution).
Hence, SQF turns out to be a promising per-flow scheduler to implement implicit service differentiation for UDP flows
with large stream rates.
\begin{table}
\centering
\begin{tabular}{|l|rr|rr|rr|}
\hline
& \multicolumn{2}{|c}{{FQ [Mbps]}} & \multicolumn{2}{|c}{{LQF [Mbps]}} &\multicolumn{2}{|c|}{{SQF [Mbps]}} \\
\hline
\hline
& $ns2$ & Model& $ns2$ & Model& $ns2$ & Model\\
& \multicolumn{6}{c|}{{UDP rate  $X^{UDP}=3Mbps<C/2$}}  \\
$\overline{L}^{UDP}$  & 0.01 & 0 &2.98  & 3 & 0.1 & 0 \\
$\overline{X}^{TCP}$  & 6.69 & 7 & 9.99& 10 & 6.97 & 7 \\
\hline
& \multicolumn{6}{c|}{{UDP rate  $X^{UDP}=7Mbps>C/2$}}  \\
$\overline{L}^{UDP}$  & 1.96 & 2 & 6.95 & 7 & 0.1 & 0 \\
$\overline{X}^{TCP}$  & 4.98 & 5 & 9.87 & 10 & 2.98 & 3 \\
\hline
\end{tabular}
\caption{Scenario TCP/UDP: $C=10$Mbit/sec, $R_1=20ms$.}
\label{tab:numerics2}
\end{table}

\begin{figure}[htbp]
\centering
\includegraphics[width=0.8\textwidth]{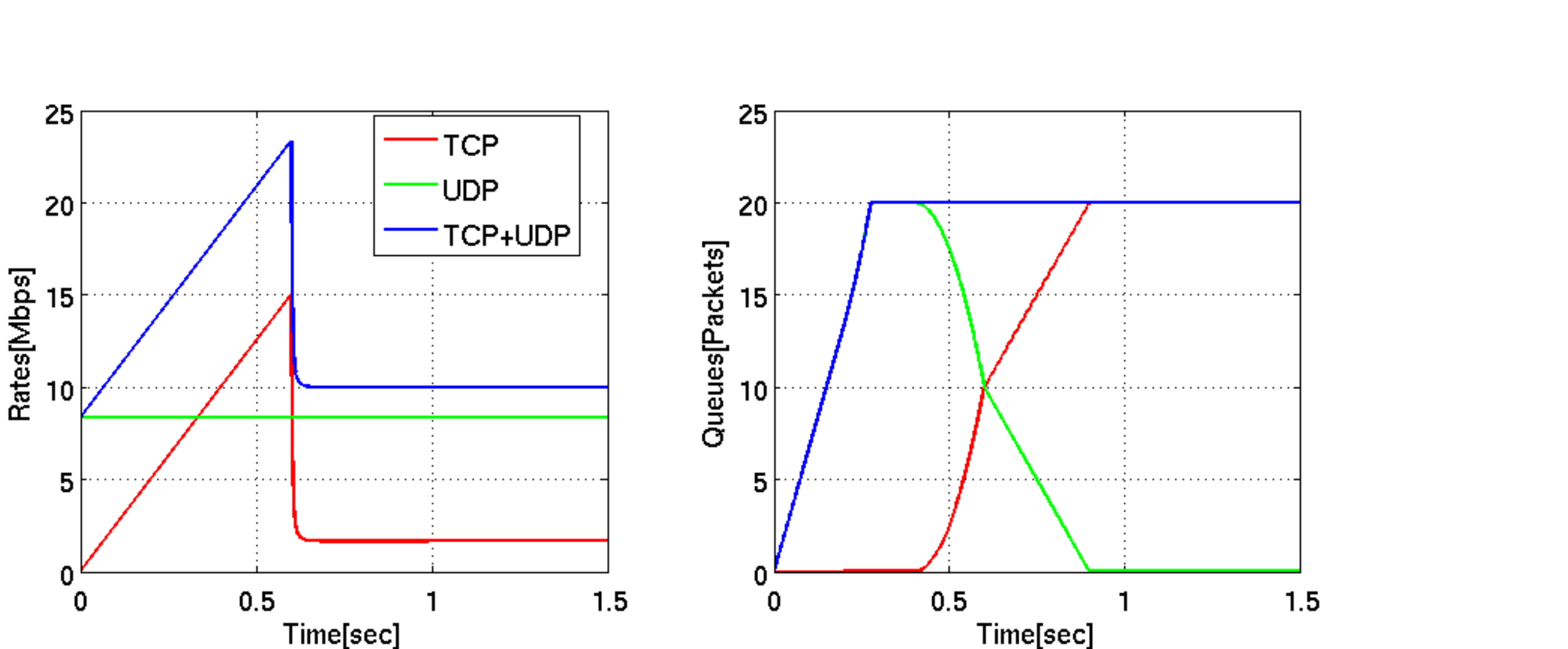}
\includegraphics[width=0.8\textwidth, height=0.25\textwidth]{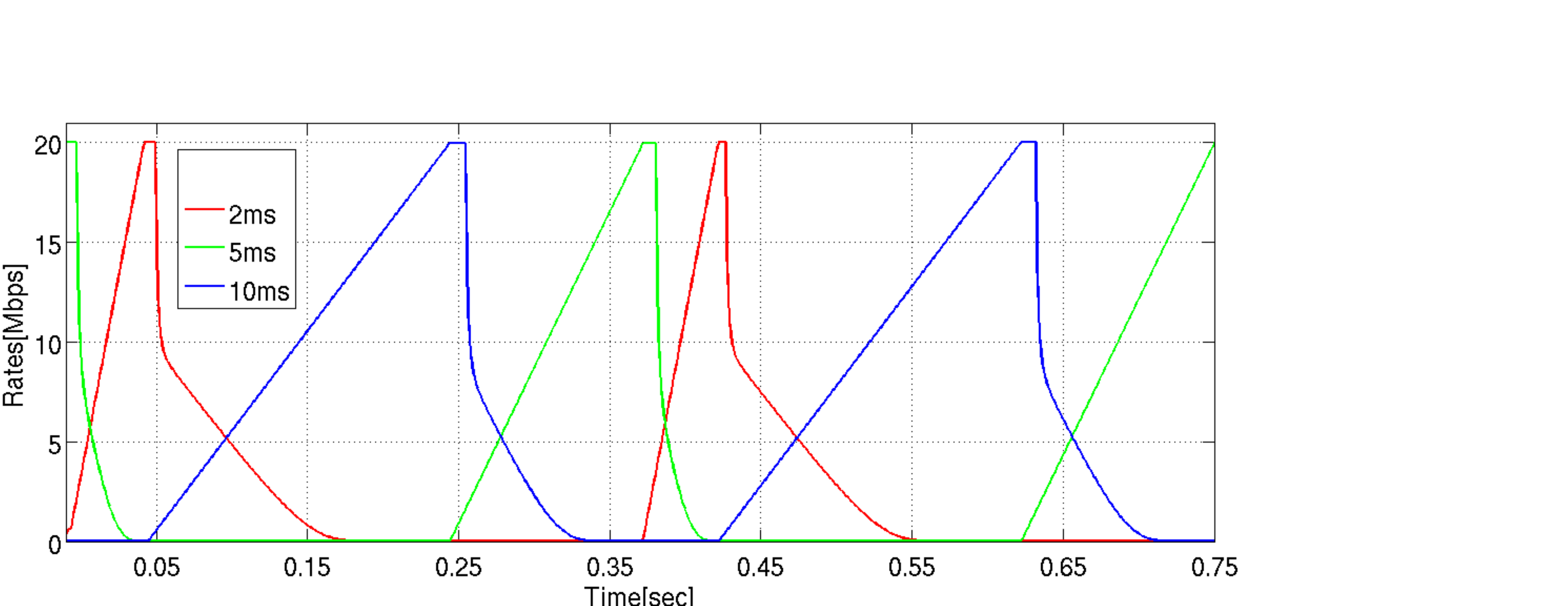}
\caption{Top plots: TCP and UDP under SQF scheduling. Bottom plot: 3 TCP flows under SQF with RTTs=2ms,5ms,10ms, B=30kB.}
\label{fig:sqf-tcp-dup}
\end{figure}
Short-term fairness in SQF can also be observed in Fig.\ref{fig:sqf-tcp-dup}, on the bottom,
where three TCP flows share a common bottleneck of $10$Mbps.
SQF allocates resources to TCP as a time division multiplexer with timeslots of variable size
proportional to RTTs. This is the reason why UDP traffic is prioritized
over TCP.
\section{Discussion and Conclusions}\label{sec:discussion}
The main contribution of this paper is given by the set of analytical
results gathered in Sec.\ref{sec:Analytical-results},\ref{sec:extensions}
that allow to capture and predict, in a relatively simple framework, the system evolution,
even in scenarios that would be difficult to study through simulation or experiments.
An additional, non marginal, contribution of the analysis is that we can now give a solid justification of many feasible
applications of such per-flow schedulers in today networks.\\
\textbf{\textit{FQ and LQF}.}
The absence of a limit-cycle for FQ and LQF in the stationary regime
explains the expected insensitivity of fairness w.r.t. the timescale.
In LQF, throughput is biased in favor of flows with small RTTs (the opposite
behavior is observed in SQF).
FQ, on the contrary, provides no biased allocations.
LQF has clearly applications in switching architectures to maximize global throughput. However,
the throughput maximization in switching fabrics in high speed routers
does not take into account the nature of Internet traffic, whose performance are closely related to TCP rate control.
As a consequence, flows that are bottlenecked elsewhere in the upstream path stall, and UDP streams
stall as well.\\
\textbf{\textit{SQF}.}
The analysis of SQF deserves particular attention. This per-flow scheduler has shown to have the attractive
property to implicitly differentiate streaming and data traffic, performing as a priority
scheduler for applications with low loss rate and delay constraints. Implicit service differentiation
is a great feature since it does not rely on the explicit knowledge of a specific application and preserves
network neutrality.
Recent literature on this subject has exploited FQ in order to achieve this goal
through the cross-protect concept (\hspace{-0.1mm}\cite{hpsr,ancs05}). Cross-protect gives priority to flows with
a rate below the fair rate and keeps the fair rate reasonably high through flow admission control
to limit the number of flows in progress. SQF achieves the same goal of cross-protect with no need of admission control as it differentiates
applications with rate much larger than the fair rate.
FQ is satisfactory on backbone links where the fair
rate can be assumed to be sufficiently high, whereas SQF has the ability to successfully replace it in access networks.
Furthermore,
SQF has a very different behavior from the other per-flow schedulers considered in this paper.
The instantaneous rate admits a limit cycle in the stationary regime, both for queues and rates,
in which flows are served in pseudo round-robin phases of duration proportional to the link capacity and to RTTs, so that long term throughput is biased in favor of flows with large RTTs.
Note that this oscillatory behavior, already observed for TCP under drop tail queue management and due to flow synchronization, here is due to the short term unfairness of SQF: the flow associated to the smallest queue is allocated
all the resources without suffering from  packet drops caused by congestion.
For instance, in LQF such phases do not exist because the flow with the longest queue is at the same time in service and penalized by
congestion packet drops.
Additional desired effects of SQF on TCP traffic, not accounted for by the model, derive
by the induced acceleration of the slow start phase (when a flow generates little queuing),
with the consequent global effect that small data transfers get prioritized on larger ones.
As a future work, we intend to investigate the effect of UDP burstiness, not captured by fluid models, on SQF implicit differentiation
capability, via a more realistic stochastic model.

\begin{appendix}
\section{Auxiliary result $1$}\label{app:1}
\begin{prop}
Given $Q(t_s)=B$ and $A(t_s)> C$, $\forall t>t_s$, the total queue remains full, i.e. $Q(t)=B$.
\end{prop}
\begin{proof}
To prove the statement is equivalent to prove that $A(t)\ge C$, $\forall t>t_s$.
In fact, under the condition $A(t)\ge C$, $\frac{dQ(t)}{dt}=A(t)-C-(A(t)-C)=0$.
Let us take eq.(\ref{eqn:ODEsource2}) and write it for $A(t), \forall t\ge t_s$:
\begin{align}\label{eqn:ODEsource3}
\frac{dA(t)}{dt}=&\sum_{k=1}^2 \frac{1}{R_k^2}\frac{D_k}{C}-\frac{A(t)}{2}(A(t)-C)^+ \I{Q(t)=B}
\end{align}
Since the derivative of $A(t)$ is lower-bounded by $(A(t)-C)^+$ and becomes zero when $A(t)=C$,  it is clear that starting from $A(t_s)>C$ the total rate will remain
greater or equal to $C$, $\forall t>t_s$, which is enough to prevent $Q$, initially full in $t_s$, from decreasing.
\end{proof}
\section{Auxiliary result $2$}\label{app:2}
The following result is the proof of the existence of the limit
the sequences $\{\epsilon_i\}$, $\{\gamma_i\}$, $\forall i\in \N,$ and by symmetry $\{\gamma_i\}$.
\begin{prop}
The positive sequences  $\{\epsilon_i\}$,  $\{\gamma_i\}$, with
\begin{enumerate}
 \item $\{\epsilon_i\}$, $\forall i  \in \N,$ defined by the recursion
$$\epsilon_{i+1}=\epsilon_{i}2\sqrt{\beta}f( \epsilon_{i}+\alpha\tau_{i},\gamma_{i},\tau'_i), \quad i>1$$
where
\begin{align*}
&\tau_i=\frac{C}{\alpha}\left(1+\sqrt{1-\frac{2\alpha\epsilon_{i-1}}{C^2}}\right), &\tau'_i=\frac{C}{\beta}\left(1+\sqrt{1-\frac{2\beta\gamma_{i-1}}{C^2}}\right);\nonumber
\end{align*}
\item $\{\gamma_i\}$, $\forall i  \in \N,$ defined by the recursion
$$\gamma_{i+1}=\gamma_{i}2\sqrt{\alpha}g(\gamma_{i}+\beta\tau'_i,\epsilon_{i},\tau_{i+1}), \quad  i>1$$
\end{enumerate}
are convergent and their limits are $\widetilde\epsilon=0$, $\widetilde\gamma=0$.
\end{prop}
\begin{proof}
Take the sequence $\{\epsilon_i\}$. It holds:
\begin{align*}
\epsilon_{i+1} &\le \left(\epsilon_{i}+\alpha\tau_i\right)2\sqrt{\beta} f(\epsilon_{i}+\alpha\tau_i,0,\tau'_i) \nonumber \\
               &\le \left(\epsilon_{i}+2C\right)\frac{2\sqrt{\beta}}{2\sqrt{\beta}+2Ce^{\frac{C^2}{2\beta}}\sqrt{2\pi} Erf(C/\sqrt{2\beta})}
\nonumber \\
           &\le\left(\epsilon_{i}+2C\right)K
\le \epsilon_1 K^{i-1}+\sum_{j=1}^{i-1} (2C)^{i-j}K^j \nonumber \\
&= \epsilon_1 K^{i-1}+(2C)^i\sum_{j=1}^{i-1}\left(\frac{K}{2C}\right)^j
\end{align*}
where we defined  $K=2\sqrt{\beta}/\left(2\sqrt{\beta}+2Ce^{\frac{C^2}{2\beta}}\sqrt{2\pi} Erf(C/\sqrt{2\beta})\right)$.\\
It follows that 
\begin{align}\label{eq:lim}
\lim_{i\rightarrow \infty} \epsilon_1 K^{i-1}+(2C)^i\sum_{j=1}^{i-1}\left(\frac{K}{2C}\right)^j =0.
\end{align}
Indeed, it results $K<1<2C$, since the link capacity satisfies $C>>1$.
Since we assumed $\epsilon_i>0$, $\forall i>1$, it implies that $\widetilde \epsilon\triangleq \lim_{i\rightarrow \infty} \epsilon_i=0$, that is the sequence converges to zero.
By symmetry, $\widetilde \gamma \triangleq \lim_{i\rightarrow \infty} \gamma_i=0$, which concludes the proof.
\end{proof}
\section{Proof of Corollary \ref{cor:SQF}}\label{app:3}
\begin{proof}
We denoted by $\overline{A}_1$ and $\overline{A}_2$ the mean values of the sending rates in steady state, that is
\begin{align}\label{def:mean_sending_rates_SQF}
&\overline{A}_1 \triangleq \frac{1}{T} \int_0^T \widetilde{A}_1(u) du, \quad \overline{A}_2 \triangleq \frac{1}{T} \int_0^T \widetilde{A}_2(u) du.
\end{align}
Decompose the integral over the limit cycle $T$ in the sum of the integral over phase $A_1^{ON}$ and $A_2^{ON}$.
It results:
\begin{align}\label{eq:mean_sending_rates_SQF}
\overline{A}_1 &= \frac{1}{T}\left( \int_0^{2C/\alpha}\alpha u du+ \int_{2C/\alpha}^{2C/\alpha+2C/\beta}\widetilde{f}(2C,0,u)du \right)\\
&\approx \frac{1}{T}\left(\frac{2C^2}{\alpha}+ \int_{2C/\alpha}^{2C/\alpha+2C/\beta} f(2C,0,u)du \right) \nonumber \\
&= \left(\frac{\beta}{\alpha+\beta}C+\frac{\alpha\beta}{2C(\alpha+\beta)}\log\left(1+\frac{2\sqrt{2\pi}}{\sqrt{\beta}}C e^{\frac{C^2}{2\beta}}Erf\left(\frac{C}{\sqrt{2\beta}}\right)\right)\right).\nonumber
\end{align}
where we approximated $\widetilde f$ with $f$.
Similarly, one gets the approximate expression for $\overline{A}_2$.
It is worth observing that the above approximation is justified by the numerical comparison of mean sending rates against $ns2$.
In addition, it is important to distinguish among TCP \textit{sending rates} $\widetilde{A}_1$, $\widetilde{A}_2$ (where $\sim$ stands for the stationary version) and throughputs, $\widetilde{X}_1$, $\widetilde{X}_2$:
during phase $A_1^{ON}$, $Q_1$ is served at capacity $C$, so $\widetilde{X}_1=C$ and $\widetilde{X}_2=0$, as all packets sent by flow $2$ are lost.
On the contrary, during phase $A_2^{ON}$, $\widetilde{X}_2=C$, whereas $\widetilde{X}_1=0$.
It results that:
\begin{align*}
\overline{X}_1 &= \frac{\alpha}{\alpha+\beta}C, &\overline{X}_2 &= \frac{\beta}{\alpha+\beta}C.
\end{align*}
In order to compute the value of the virtual queues in steady state, it suffices to recall that $\widetilde{Q}_1(t)=B-\widetilde{Q}_2(t)$ and over phase $A_k^{ON}$, $\widetilde{Q}_k=\frac{B}{2}+\int (A_k(t)-C)dt=\frac{B}{2}+\int (\alpha t-C)dt$.
Hence,
\begin{align}\label{mean-queue1}
\overline{Q}_1&=\frac{1}{T}\left(\int_{\text{phase}A_1^{ON}} \left(\frac{B}{2}+\widetilde{Q}_1(t)\right) dt + 
\int_{\text{phase}A_2^{ON}}(B-\widetilde{Q}_2(t)) dt \right)\nonumber \\
&= \frac{B}{2}+\frac{C^2(\alpha-\beta)}{3\alpha\beta},
\end{align}
where $T$ denotes the limit-cycle duration, i.e. $T=2C(\frac{1}{\alpha}+\frac{1}{\beta})$.
Similarly, one gets
\begin{align}\label{mean-queue2}
\overline{Q}_2&=\frac{B}{2}-\frac{C^2(\alpha-\beta)}{3\alpha\beta}.
\end{align}
Under the assumption $\alpha>\beta$, it results that, as expected, $\overline{Q}_1>B/2$, while $\overline{Q}_2<B/2$.\\
Note in addition that the approximation $\widetilde{f}(\cdot) \approx f(\cdot)$ in (\ref{eq:mean_sending_rates_SQF}) has no impact neither on the mean throughput nor on the virtual queue formulas.

\end{proof}

\end{appendix}

\end{document}